\documentclass[a4paper,10pt]{article}
\usepackage[utf8]{inputenc}
\usepackage{amsthm}
\usepackage{amsmath}
\usepackage{authblk}
\usepackage{mathtools}
\usepackage[mathscr]{euscript}
\usepackage[pdftex]{graphicx}
   \usepackage[justification=centering]{caption}
   \usepackage{float}
\usepackage{cite}
\usepackage{amssymb}
\usepackage{hyperref}
\newtheorem{theorem}{Theorem}

\newtheorem{lemma}{Lemma}
\newtheorem{corollary}{Corollary}

\newtheorem{remark}{Remark}
\title{Analysis of gradient descent methods with non-diminishing, bounded errors}
\author{Arunselvan Ramaswamy
\thanks{email:\texttt{arunselvan@iisc.ac.in}}} 
\affil{Dept. of Computer Science and Automation \\
Indian Institute of Science\\
Bengaluru - 560012, India.}

\author{Shalabh Bhatnagar
\thanks{email:\texttt{shalabh@iisc.ac.in}}}
\affil{Dept. of Computer Science and Automation \\
Indian Institute of Science\\
Bengaluru - 560012, India.
}
\begin{document}
\maketitle
\begin{abstract}
The main aim of this paper is to provide an analysis of gradient descent ($GD$) algorithms
with gradient errors that do not necessarily vanish, asymptotically.
In particular, sufficient conditions are presented for both stability (almost sure boundedness of
the iterates) and convergence of $GD$ with bounded, (possibly) non-diminishing gradient errors. In addition
to ensuring stability, such an algorithm is shown to 
converge to a small neighborhood of the minimum set,
which depends on the gradient errors.
It is worth noting that the main result
of this paper can be used to show that $GD$ with asymptotically vanishing errors indeed converges 
to the
minimum set. 
The results presented herein are not only more general when compared to previous
results, but our analysis of \textit{$GD$ with errors} is new to the literature to the best of our
knowledge.
Our work extends the contributions of Mangasarian \& Solodov,
Bertsekas \& Tsitsiklis and Tadi{\'c} \& Doucet. 
Using our framework, a simple yet effective implementation of $GD$ using simultaneous
perturbation stochastic approximations ($SPSA$), with constant
sensitivity parameters, is presented. 
Another important improvement over many previous results is that there are
no `additional' restrictions imposed on the step-sizes. In machine learning applications
where step-sizes are related
to learning rates, our assumptions, unlike those of other papers, do not affect these learning
rates. Finally,
we present experimental results to validate our theory.
\end{abstract}
\section{Introduction} \label{sec:introduction}
Let us suppose that we are interested in finding a minimum (local/global) of
a continuously differentiable function $f: \mathbb{R}^d \to \mathbb{R}$.
The following gradient descent method ($GD$) is
often employed to find such a minimum:
\begin{equation} \label{eq:GD}
 x_{n+1} = x_n - \gamma (n) \nabla f(x_n).
\end{equation}
In the above equation, $\{\gamma(n)\}_{n \ge 0}$ is the given step-size sequence and
$\nabla f : \mathbb{R}^d \to \mathbb{R}^d$ is a continuous map such that 
$\lVert \nabla f(x) \rVert \le K(1 + \lVert x \rVert)$, $K > 0$ and $x \in \mathbb{R}^d$. 
$GD$ is a popular tool to implement many machine
learning algorithms. For example, the backpropagation algorithm for training neural networks employs 
$GD$ due to its effectiveness and ease of implementation.

When implementing (\ref{eq:GD}), one often uses gradient estimators such as
Kiefer-wolfowitz estimator \cite{kiefer}, simultaneous perturbation stochastic approximation
($SPSA$) \cite{spall}, etc., to obtain estimates of the true gradient at each stage
which in turn results in \textit{estimation errors} ($\epsilon_n$ in (\ref{eq:BP1})).
This is particularly true when the form of $f$ or $\nabla f$ is unknown.
Previously in the literature, convergence of \textit{$GD$ with errors} 
was studied in \cite{bertsekas2000gradient}.
However, their analysis required the errors to go to zero at the rate of the step-size (vanish asymptotically 
at a prescribed rate).
Such assumptions are difficult to enforce and may adversely affect the learning rate 
when employed to implement machine learning algorithms, see
Chapter 4.4 of \cite{haykin}.
In this paper, we present sufficient conditions for both stability (almost sure boundedness)
and convergence (to a small neighborhood of the minimum set) of $GD$ with bounded errors, for which the recursion is given by
\begin{equation} \label{eq:BP1}
x_{n+1} = x_n - \gamma_n (\nabla f(x_n) + \epsilon_n).
\end{equation}
In the above equation $\epsilon_n$ is the estimation error 
at stage $n$ such that $\forall n$ $\lVert \epsilon_n \rVert \le \epsilon$ (a.s. in the case
of stochastic errors) for a fixed $\epsilon > 0$ (positive real).
As an example, consider the problem of estimating the average waiting time of a customer in a queue.
The objective function $J$, for this problem, has the following form:
$
 J(x) = \int w \ d(F(w \mid x)) = E[W(x)],
$
where $W(x)$ is the ``waiting time'' random variable with distribution $F(\cdotp \mid x)$,
with $x$ being the underlying parameter (say the arrival or the service rate). 
In order to define $J$ at every $x$,
one will need to know the entire family of distributions, $\{ F(\cdotp \mid x) \mid x \in \mathbb{R}^d\}$, exactly.
In such scenarios, one often works with approximate definitions of $F$ which in turn lead to
approximate gradients, \textit{i.e,} gradients with errors. More generally,
the gradient errors could be inherent to the problem at hand or due to extraneous noise.
In such cases, there is no reason to believe that these errors will vanish asymptotically.
\textit{To the best of our knowledge, this is the first time 
an analysis is done for $GD$ with biased/unbiased stochastic/deterministic errors that are not necessarily diminishing, and
without imposing `additional' restrictions
on step-sizes over the usual standard assumptions}, see (A2) in Section~\ref{sec:assumptions}. 

Our assumptions, see Section~\ref{sec:assumptions}, not only guarantee stability but also guarantee
convergence of the algorithm to a small neighborhood of the minimum set, where the neighborhood
is a function of the gradient errors. If $\lVert \epsilon_n \rVert \to 0$ as $n \to \infty$,
then it follows from our main result (Theorem~\ref{corr}) that the algorithm converges to an
arbitrarily small neighborhood of the minimum set. \textit{In other words, the algorithm indeed 
converges to the minimum set}.
It may be noted that we do not impose any restrictions on the noise-sequence $\{\epsilon_n\}_{n \ge 0}$,
except that almost surely for all $ n$ $\lVert \epsilon_n \rVert \le \epsilon$ for some fixed
$\epsilon > 0$.
Our analysis uses techniques developed in the field of 
viability theory by \cite{Aubin}, \cite{Benaim05} and \cite{Benaim12}.
Experimental results supporting the analyses in this paper are presented in Section~\ref{experiments}.
\subsection{Our contributions}

 \textbf{(1)} Previous literature such as \cite{bertsekas2000gradient} requires
 $\lVert \epsilon_n \rVert \to 0$ as $n \to \infty$ for it's analysis to work. 
 Further, both \cite{bertsekas2000gradient}
  and \cite{mangasarian} provide conditions that guarantee one of two things $(a)$
  $GD$ diverges almost surely or $(b)$ converges to the minimum set almost surely.
  On the other hand, we
  only require $\lVert \epsilon_n \rVert \le \epsilon$ $\forall \ n$, where $\epsilon > 0$ is fixed a priori.
  Also, we present conditions under which $GD$ with bounded errors is stable 
  (bounded almost surely) and converges to an arbitrarily small
  neighborhood of the minimum set almost surely. 
  \textit{Note that our analysis works regardless of whether or not $\lVert \epsilon_n \rVert$ tends
  to zero.}
  For more detailed comparisons with \cite{bertsekas2000gradient}
  and \cite{mangasarian} see Section~\ref{comparisons}.
  \\ \textbf{(2)} The analyses presented herein will go through even when the gradient errors are ``asymptotically
  bounded'' almost surely. In other words, $\lVert \epsilon_n \rVert \le \epsilon$ for all 
  $n \ge N$ almost surely.
  \textit{Here $N$ may be sample path dependent}.
  \\ \textbf{(3)} Previously, convergence analysis of $GD$ required severe restrictions on the step-size,
  see \cite{bertsekas2000gradient}, \cite{spall}. However, in our paper we do not impose any
  such restrictions on the step-size. See Section~\ref{comparisons}
  (specifically points $1$ and $3$) for more details.
  \\ \textbf{(4)} Informally, the main result of our paper, Theorem~\ref{corr}, states the following. 
  \textit{One wishes to simulate $GD$ with gradient errors that are not
  guaranteed to vanish over time. As a consequence of allowing non-diminishing errors, 
  we show the following: There exists $\epsilon(\delta) > 0$ such that the iterates
  are stable and converge to the $\delta$-neighborhood of the minimum set ($\delta$ being chosen by the simulator)
  as long as $\lVert \epsilon_n \rVert \le \epsilon(\delta)$ $\forall \ n$.
  }
  \\ \textbf{(5)} In Section~\ref{implementBP} we discuss how our framework can be exploited to \textit{undertake
  convenient yet effective implementations of }$GD$. Specifically, we present an implementation using $SPSA$,
  although other implementations can be similarly undertaken. 
\section{Definitions used in this paper}\label{sec:definitions}
\noindent
\textit{\textbf{[Minimum set of a function]}} This set consists of all global and local minima of
the given function.
\\
 \textit{\textbf{[Upper-semicontinuous map]}} We say that $H$ is upper-semicontinuous,
  if given sequences $\{ x_{n} \}_{n \ge 1}$ (in $\mathbb{R}^{n}$) and 
  $\{ y_{n} \}_{n \ge 1}$ (in $\mathbb{R}^{m}$)  with
  $x_{n} \to x$, $y_{n} \to y$ and $y_{n} \in H(x_{n})$, $n \ge 1$, 
  then $y \in H(x)$.
\\
\textit{\textbf{[Marchaud Map]}} A set-valued map $H: \mathbb{R}^n \to \{subsets\ of\ \mathbb{R}^m \}$ 
is called \textit{Marchaud} if it satisfies
the following properties:
 \textbf{(i)} for each $x$ $\in \mathbb{R}^{n}$, $H(x)$ is convex and compact;
 \textbf{(ii)} \textit{(point-wise boundedness)} for each $x \in \mathbb{R}^{n}$,  
 $\underset{w \in H(x)}{\sup}$ $\lVert w \rVert$
 $< K \left( 1 + \lVert x \rVert \right)$ for some $K > 0$;
 \textbf{(iii)} $H$ is \textit{upper-semicontinuous}. \\
Let $H$ be a Marchaud map on $\mathbb{R}^d$.
The differential inclusion (DI) given by
\begin{equation} \label{di}
\dot{x} \ \in \ H(x)
\end{equation}
is guaranteed to have at least one solution that is absolutely continuous. 
The reader is referred to \cite{Aubin} for more details.
We say that $\textbf{x} \in \sum$ if $\textbf{x}$ 
is an absolutely continuous map that satisfies (\ref{di}).
The \textit{set-valued semiflow}
$\Phi$ associated with (\ref{di}) is defined on $[0, + \infty) \times \mathbb{R}^d$ as: \\
$\Phi_t(x) = \{\textbf{x}(t) \ | \ \textbf{x} \in \sum , \textbf{x}(0) = x \}$. Let
$B \times M \subset [0, + \infty) \times \mathbb{R}^d$ and define
$
 \Phi_B(M) = \underset{t\in B,\ x \in M}{\bigcup} \Phi_t (x).
$
\\
\textit{\textbf{[Limit set of a solution]}} The limit set of a solution $\textbf{x}$
with $\textbf{x}(0) = x$ is given by
$L(x) = \bigcap_{t \ge 0} \ \overline{\textbf{x}([t, +\infty))}$.
\\
\textit{\textbf{[Invariant set]}}
$M \subseteq \mathbb{R}^d$ is \textit{invariant} if for every $x \in M$ there exists 
a trajectory, $\textbf{x} \in \sum$, entirely in $M$
with $\textbf{x}(0) = x$, $\textbf{x}(t) \in M$,
for all $t \ge 0$.
\\ \textit{\textbf{[Open and closed neighborhoods of a set]}}
Let $x \in \mathbb{R}^d$ and $A \subseteq \mathbb{R}^d$, then
$d(x, A) : = \inf \{\lVert a- y \rVert \ | \ y \in A\}$. We define the $\delta$-\textit{open neighborhood}
of $A$ by $N^\delta (A) := \{x \ |\ d(x,A) < \delta \}$. The 
$\delta$-\textit{closed neighborhood} of $A$ 
is defined by $\overline{N^\delta} (A) := \{x \ |\ d(x,A) \le \delta \}$.
\\ \textit{\textbf{[$B_r(0)$ and $\overline{B}_r(0)$]}}
The open ball of radius $r$ around the origin is represented by $B_r(0)$,
while the closed ball is represented by $\overline{B}_r(0)$. In other words,
$B_r(0) = \{x \mid \lVert x \rVert < r\}$ and $\overline{B}_r(0) = \{x \mid \lVert x \rVert \le r\}$.
\\ \textit{\textbf{[Internally chain transitive set]}}
$M \subset \mathbb{R}^{d}$ is said to be
internally chain transitive if $M$ is compact and for every $x, y \in M$,
$\epsilon >0$ and $T > 0$ we have the following: There exists $n$ and $\Phi^{1}, \ldots, \Phi^{n}$ that
are $n$ solutions to the differential inclusion $\dot{x}(t) \in h(x(t))$,
points $x_1(=x), \ldots, x_{n+1} (=y) \in M$
and $n$ real numbers 
$t_{1}, t_{2}, \ldots, t_{n}$ greater than $T$ such that: $\Phi^i_{t_{i}}(x_i) \in N^\epsilon(x_{i+1})$ and
$\Phi^{i}_{[0, t_{i}]}(x_i) \subset M$ for $1 \le i \le n$. The sequence $(x_{1}(=x), \ldots, x_{n+1}(=y))$
is called an $(\epsilon, T)$ chain in $M$ from $x$ to $y$. If the above property only holds for all $x=y$,
then $M$ is called \textit{\textbf{chain recurrent}}.
\\ \textit{\textbf{[Attracting set \& fundamental neighborhood]}}
$A \subseteq \mathbb{R}^d$ is \textit{attracting} if it is compact
and there exists a neighborhood $U$ such that for any $\epsilon > 0$,
$\exists \ T(\epsilon) \ge 0$ with $\Phi_{[T(\epsilon), +\infty)}(U) \subset
N^{\epsilon}(A)$. Such a $U$ is called the \textit{fundamental neighborhood} of $A$.
\\ \textit{\textbf{[Attractor set]}}
An \textit{attracting set} that is also invariant
is called an \textit{attractor set}.
The \textit{basin
of attraction } of $A$ is given by $B(A) = \{x \ | \ \omega_\Phi(x) \subset A\}$.
\\
\textit{\textbf{[Lyapunov stable]}} The above set $A$ is Lyapunov stable 
if for all $\delta > 0$, $\exists \ \epsilon > 0$ such that
$\Phi_{[0, +\infty)}(N^\epsilon(A)) \subseteq N^\delta(A)$.
\\ \textit{\textbf{[Upper-limit of a sequence of sets, Limsup]}}
Let $\{K_{n}\}_{n \ge 1}$ be a sequence of sets in $\mathbb{R}^{d}$. 
The \textit{upper-limit} of $\{K_{n}\}_{n \ge 1}$ is
given by, 
$Limsup_{n \to \infty} K_n$ $:= \ \{y \ | \ 
\underset{n \to \infty}{\underline{lim}}d(y, K_n)= 0 \}$. \\
We may interpret that the lower-limit collects the limit points of  
$\{K_n\}_{n \ge 1}$ while the upper-limit
collects its accumulation points.
\section{Assumptions and comparison to previous literature} 
\subsection{Assumptions}\label{sec:assumptions}
Recall that $GD$ with bounded errors is given by the following recursion:
\begin{equation} \label{eq:BP2}
  x_{n+1} = x_n - \gamma(n) g(x_n) ,
\end{equation}
where $g(x_n) \in G(x_n)$ $\forall \ n$ and $G(x) := \nabla f(x) + \overline{B}_\epsilon (0)$, $x \in \mathbb{R}^d$.
In other words, the gradient estimate at stage $n$, $g(x_n)$, belongs to an $\epsilon$-ball 
around the true gradient $\nabla f(x_n)$ at stage $n$. 
Note that (\ref{eq:BP2}) is consistent with (\ref{eq:BP1}) of Section~\ref{sec:introduction}.
Our assumptions, $(A1)$-$(A4)$ are listed below.
\begin{itemize}
 \item[\textbf{(A1)}]  $G(x) := \nabla f(x) + \overline{B}_\epsilon (0)$ for some fixed $\epsilon > 0$.
$\nabla f$ is a continuous function such that 
$\lVert \nabla f(x) \rVert \le K(1 + \lVert x \rVert)$ for all $x \in \mathbb{R}^d$, for some $K > 0$.
\item[\textbf{(A2)}] $\{ \gamma(n) \}_{n \ge 0}$ is the step-size (learning rate) sequence 
such that: $\gamma(n) > 0$ $\forall n$, 
 $\underset{n \ge 0}{\sum} \gamma(n) = \infty$
 and 
 $\underset{n \ge 0}{\sum} \gamma(n)^{2} < \infty$. Without loss of generality we let
 $\underset{n}{sup}\ \gamma(n) \le 1$.  
\end{itemize}
Note that $G$ is an upper-semicontinuous map since $\nabla f$ is continuous and point-wise bounded.
For each $c \ge 1$, we define $G_c (x) := \{ y/c \mid y \in G(cx) \}$.
Define $G_\infty (x) := \overline{co} \left\langle Limsup_{c \to \infty} G_c (x) \right\rangle$,
see Section~\ref{sec:definitions} for the definition of $Limsup$. Given
$S \subseteq \mathbb{R}^d$, the convex closure of $S$, denoted by $\overline{co} \langle S \rangle$,
is the closure of the convex hull of $S$. It is worth noting that
$Limsup_{c \to \infty} G_c (x)$ is non-empty for every $x \in \mathbb{R}^d$. Further, we
show that $G_\infty$ is a Marchaud map in Lemma~\ref{ginfmarchaud}. In other words,
$\dot{x}(t) \in -G_\infty (x(t))$ has at least one solution that is absolutely continuous, see \cite{Aubin}.
Here $-G_{\infty}(x(t))$ is used to denote the set $\{-g \mid g \in G_\infty(x(t))\}$.
\begin{itemize}
 \item[\textbf{(A3)}] $\dot{x}(t) \in -G_\infty (x (t))$ has an attractor set $\mathcal{A}$ such 
that $\mathcal{A} \subseteq B_a (0)$ for some $a > 0$ and $\overline{B}_a (0)$ is a fundamental neighborhood of
$\mathcal{A}$.
\end{itemize}
Since $\mathcal{A} \subseteq B_a (0)$ is compact, we have that
$\underset{x \in \mathcal{A}}{\sup} \lVert x \rVert < a$.
Let us fix the following sequence of real
numbers: $\underset{x \in \mathcal{A}}{\sup} \lVert x \rVert = \delta_1 < \delta_2
< \delta_3 < \delta_4 < a$.
\begin{itemize}
\item[\textbf{(A4)}] Let $c_{n} \ge 1$ be an increasing sequence of integers such that
$c_{n} \uparrow \infty$ as $n \to \infty$. Further, let
$x_{n} \ \rightarrow \ x$ and $y_{n} \ \rightarrow \ y$ as 
$n \ \rightarrow \infty$, such that $y_{n} \in G_{c_{n}}(x_{n})$, $\forall n$, 
then $y \in G_{\infty}(x)$.
\end{itemize}
\textit{\textbf{
It is worth noting that the 
existence of a global Lyapunov function for $\dot{x}(t) \in -G_\infty (x(t))$
is sufficient to guarantee that $(A3)$ holds. Further, $(A4)$ is satisfied when $\nabla f$ is Lipschitz continuous.}}
\begin{lemma} \label{ginfmarchaud}
 $G_\infty$ is a Marchaud map.
\end{lemma}
\begin{proof}
 From the definition of $G_\infty$ and $G$ we have that $G_\infty (x)$ is convex, compact
and $\underset{y \in G(x)}{\sup} \lVert y \rVert \le K(1 + \lVert x \rVert)$ for every $x \in \mathbb{R}^d$.
It is left to show that $G_\infty$ is an upper-semicontinuous map.
Let $x_{n} \to x$, $y_{n} \to y$ and $y_{n} \in G_{\infty}(x _n)$, 
 for all $n \ge 1$.
 We need to show that $y \in G_{\infty}(x)$. We present a proof by contradiction.
 Since $G_{\infty}(x)$ is convex and compact, $y \notin G_{\infty}(x)$
 implies that there exists a linear functional on $\mathbb{R}^{d}$, say $f$, such that
 $\underset{z \in G_{\infty}(x)}{sup}$ $f(z) \le \alpha - \epsilon$
 and $f(y) \ge \alpha + \epsilon$, for some
 $\alpha \in \mathbb{R}$ and $\epsilon > 0$. Since $y_{n} \to y$, there exists 
 $N > 0$ such that for all $n \ge N$, $f(y_{n}) \ge \alpha + \frac{\epsilon}{2}$. In other
 words, $G_{\infty}(x) \cap  [f \ge \alpha + \frac{\epsilon}{2}] \neq \phi$ for
 all $n \ge N$. We use the notation $[f \ge a]$ to denote the set
 $\left\{ x \ |\ f(x) \ge a \right\}$. For the sake of convenience let us denote the set
 $Limsup_{c \to \infty}G_{c}(x)$ by $A(x)$, where $x \in \mathbb{R}^{d}$.
 We claim that $A(x_{n}) \cap [f \ge \alpha + \frac{\epsilon}{2}] \neq \phi$
 for all $n \ge N$. We prove this claim later,
 for now we assume that the claim
 is true and proceed. Pick $z_{n} \in A(x_{n}) \cap [f \ge \alpha + \frac{\epsilon}{2}]$
 for each $n \ge N$. It can be shown that $\{z_{n}\}_{n \ge N}$ is norm bounded
 and hence contains a convergent subsequence, 
 $\{z_{n(k)}\}_{k \ge 1} \subseteq \{z_{n}\}_{n \ge N}$. 
 Let $\underset{k \to \infty}{\lim} z_{n(k)} = z$.
Since $z_{n(k)} \in Limsup_{c \to \infty}(G_{c}(x_{n(k)}))$, 
 $\exists$ $c_{n(k)} \in \mathbb{N}$ such that $\lVert w_{n(k)} - z_{n(k)} \rVert
 < \frac{1}{n(k)}$, where $w_{n(k)} \in G_{c_{n(k)}}(x_{n(k)})$. 
 We choose the sequence
 $\{c_{n(k)}\}_{k \ge 1}$ such that $c_{n(k+1)} > c_{n(k)}$ for each $k \ge 1$.
 \\ \indent
 We have the following: $c_{n(k)} \uparrow \infty$, $x_{n(k)} \to x$, 
 $w_{n(k)} \to z$ and $w_{n(k)} \in G_{c_{n(k)}}(x_{n(k)})$, for all $ k \ge 1$. 
 It follows
 from assumption $(A4)$ that $z \in G_{\infty}(x)$. Since $z_{n(k)} \to z$
 and $f(z_{n(k)}) \ge \alpha + \frac{\epsilon}{2}$ for each $k \ge 1$, we have that
 $f(z) \ge \alpha + \frac{\epsilon}{2}$. This contradicts the earlier conclusion that
 $\underset{z \in h_{\infty}(x)}{sup}$ $f(z) \le \alpha - \epsilon$.
 \\ \indent
 It remains to prove that  $A(x_{n}) \cap [f \ge \alpha + \frac{\epsilon}{2}] \neq \phi$
 for all $n \ge N$. If this were not true, then
 $\exists \{m(k)\}_{k \ge 1} \subseteq \{n \ge N\}$ 
 such that $A(x_{m(k)}) \subseteq [f < \alpha + \frac{\epsilon}{2}]$
 for all $k$. It follows that
$G_\infty(x_{m(k)}) = \overline{co}(A(x_{m(k)})) \subseteq 
 [f \le \alpha + \frac{\epsilon}{2}]$ for each $k \ge 1$. 
 Since $y_{n(k)} \to y$, $\exists N_{1}$ such that for all $n(k) \ge N_1$, 
 $f(y_{n(k)}) \ge \alpha + \frac{3 \epsilon}{4}$. This is a contradiction.
\end{proof}
 \subsection{Relevance of our results} \label{comparisons}
 \noindent
  \textbf{(1)} Gradient algorithms with errors have been previously studied by 
  Bertsekas and Tsitsiklis \cite{bertsekas2000gradient}. They impose the following restriction on the 
 estimation errors: $\lVert \epsilon_n \rVert \le \gamma(n)(q + p \lVert \nabla f(x_n) \rVert)$
 $\forall \ n$, where $p, q > 0$. If the iterates are stable then $\lVert \epsilon_n \rVert
 \to 0$. In order to satisfy the aforementioned assumption the choice of step-size 
 may be restricted, thereby affecting the learning rate (when used within the framework of a learning algorithm). 
 In this paper we analyze the more
 general and practical case of bounded $\lVert \epsilon_n \rVert$ which does not necessarily 
 go to zero. Further \textit{none} of the assumptions used in our paper impose further restrictions 
 on the step-size, other than standard requirements, see $(A2)$. \\
 \textbf{(2)} The main result of Bertsekas and Tsitsiklis \cite{bertsekas2000gradient} states that
 the $GD$ with errors either diverges almost surely or converges to the minimum set almost surely.
 An older study by Mangasarian and Solodov \cite{mangasarian} 
 shows the exact same result as \cite{bertsekas2000gradient} but for
 $GD$ without estimation errors ($\epsilon_n = 0 \ \forall \ n$). The main results of our
 paper, Theorems~\ref{stability} \& \ref{corr} show that if the $GD$  under
 consideration satisfies $(A1)$-$(A4)$ then the iterates are stable (bounded almost surely).
 Further, the algorithm is guaranteed to converge to a \textit{given 
 small neighborhood of the minimum set} provided the estimation errors are bounded by
 a constant that is a function of the neighborhood size.
 To summarize, under the more restrictive setting of \cite{bertsekas2000gradient}
 and \cite{mangasarian} the $GD$ is \textit{not} guaranteed to be stable, 
 see the aforementioned references, while the assumptions
 used in our paper are less restrictive and guarantee stability
 under the more general setting of bounded error $GD$. \textit{It may also be noted that $\nabla f$
 is assumed to be Lipschitz continuous by \cite{bertsekas2000gradient}. This turns out to be sufficient
 (but not necessary)
 for $(A1)$ \& $(A4)$ to be satisfied}.\\
 \textbf{(3)} The analysis of Spall \cite{spall} can be used to analyze a variant of $GD$ that uses $SPSA$
 as the gradient estimator. Spall introduces a gradient sensitivity parameter $c_n$ in order
 to control the estimation error $\epsilon_n$ at stage $n$. It is assumed that
 $c_n \to 0$ and $\sum_{n \ge 0} \left( \frac{\gamma(n)}{c_n} \right)^2 < \infty$,
 see A1, Section III, \cite{spall}. Again, this restricts the choice of step-size
 and affects the learning rate.
 In this setting our analysis works for the more practical scenario where $c_n = c$ for all $n$
 \textit{i.e.,} a constant, see Section~\ref{implementBP}.
 \\
 \textbf{(4)} The important advancements of this paper are the following: (i) Our framework is more general and practical since
 the errors are not required to go to zero;
 (ii) We provide easily verifiable, 
 non-restrictive set of assumptions that ensure almost sure boundedness and convergence
 of $GD$ and (iii) Our assumptions $(A1)$-$(A4)$ do not affect the choice of step-size.
 \\
 \textbf{(5)}
 Tadi{\'c} and Doucet \cite{tadic} showed that GD with bounded
 non-diminishing errors converges to a small neighborhood of the minimum set. 
 They make the following key assumption:
 \textbf{(A)} There exists $p \in (0,1]$, such that for every compact set $Q \subset \mathbb{R}^d$
 and every $\epsilon \in [0, \infty)$, $m(A_{Q,\epsilon}) \le M_Q \epsilon^p$, where 
 $A_{Q,\epsilon} = \{ f(x) \mid x \in Q, \lVert f(x) \rVert \le \epsilon \}$ and $M_Q \in [1, \infty)$.
 
 Note that $m(A)$ is the Lebesgue measure of the set $A \subset \mathbb{R}^d$. The above assumption holds if $f$
 is $d_0$ times differentiable, where $d < d_0 < \infty$, see \cite{tadic} for details. In comparison,
 we only require that the chain recurrent set of $f$ be a subset of it's minimum set. One sufficient condition
 for this is given in \textit{Proposition 4} of Hurley \cite{hurley}.
 \begin{remark} \label{benaimremark}
 Suppose the minimum set $\mathcal{M}$ of $f$, contains the chain recurrent set of 
 $\dot{x}(t) = -\nabla f(x(t))$, then it can be shown that GD without errors ($\epsilon = 0$ in
 (\ref{eq:BP2})) will converge to $\mathcal{M}$
 almost surely, see \cite{Benaim96}. On the other hand suppose
there are chain recurrent points outside $\mathcal{M}$, it may converge
to this subset (of the chain recurrent set) outside $\mathcal{M}$.
In Theorem~\ref{corr}, we will use the upper-semicontinuity of chain recurrent sets (Theorem 3.1
of Bena\"{i}m, Hofbauer and Sorin \cite{Benaim12}), to show that GD with errors will converge to a small 
neighborhood of the limiting set of the ``corresponding GD without errors''. 
In other words, GD with errors converges
to a small neighborhood of the minimum set provided the corresponding GD without errors converges to the minimum set. 
This will trivially happen if the chain recurrent set of $\dot{x}(t) = -\nabla f(x(t))$ is a subset
of the minimum set of $f$, which we implicitly assume is true. 
Suppose GD without errors does not converge to the minimum set, then it is reasonable 
to expect that GD with errors may not converge to a small neighborhood of the minimum set.

Suppose $f$ is continuously differentiable and it's regular values
(\textit{i.e.,} $x$ for which $\nabla f(x) \neq 0$) are dense in $\mathbb{R}^d$,
then the chain recurrent set of $f$ is a subset of it's minimum set,
see \textit{Proposition 4} of Hurley \cite{hurley}.
We implicitly assume that an assumption of this kind is satisfied.
\end{remark}
\section{Proof of stability and convergence} \label{sec:traj}
We use (\ref{eq:BP2}) to construct the linearly interpolated trajectory, $\overline{x}(t)$
for $t \in [0, \infty)$. First, define 
$t(0) := 0$ and $t(n) \ := \ \sum_{i=0}^{n-1} \gamma (i)$ for $n \ge 1$. Then, define
$\overline{x}(t(n)) \ := \ x_{n}$ and for 
$t \ \in \ [t(n), t(n+1)]$,
$\overline{x}(t)$ is the continuous linear interpolation of $\overline{x}(t_n)$
and $\overline{x}(t_{n+1})$.
We also construct the following piece-wise constant trajectory $\overline{g}(t)$, $t \ge 0$
as follows: $\overline{g}(t) := g(x_{n})$ for $t \in [t(n), t(n+1))$, $n \ge 0$.

We need to divide time, $[0, \infty)$, into intervals of length $T$, where
$T = T(\delta_2 - \delta_1) + 1$.
Note that $T(\delta_2 - \delta_1)$ is such that
$\Phi_t (x_0) \in N^{\delta_2 - \delta_1} (\mathcal{A})$ for $t \ge T(\delta_2 - \delta_1)$, where 
$\Phi_t (x_0)$ denotes solution to 
$\dot{x}(t) \in G_\infty (x(t))$ at time $t$ with initial condition $x_0$ and
$x_0 \in \overline{B}_a (0)$. Note that $T(\delta_2 - \delta_1)$ is independent of the
initial condtion $x_0$, see Section~\ref{sec:definitions} for more details.
Dividing time is done as follows:
define $T_{0} \ := \ 0$ and $T_{n} \ := \ min\{ t(m) \ : \ t(m) \ge T_{n-1}+T\}$, $n \ge 1$.
Clearly, there exists a subsequence $\{t(m(n))\}_{n \ge 0}$ of $\{t(n)\}_{n \ge 0}$ 
such that $T_{n} = t(m(n))$ $\forall \ n \ge 0$.
In what follows we use $t(m(n))$ and $T_n$ interchangeably.

To show stability, we use a projective scheme where the iterates are projected 
periodically, with period $T$, onto the closed ball
of radius $a$ around the origin, $\overline{B}_a (0)$. 
Here, the radius $a$ is given by $(A3)$. This projective scheme gives rise to the following
rescaled trajectories $\hat{x} (\cdotp)$ and $\hat{g} (\cdotp)$.
First, we construct $\hat{x}(t)$, $t \ge 0$: Let
$t \in [T_{n}, T_{n+1})$ for some $n \ge 0$,
then
$\hat{x}(t) := \frac{\overline{x}(t)}{r(n)}$, where 
$r(n) = \frac{\lVert\overline{x}(T_{n}) \rVert}{a} \vee 1$ ($a$ is defined in $(A3)$).
Also, let $\hat{x}(T_{n+1}^{-})\ := $  
$\underset{t \uparrow T_{n+1}}{\lim} \hat{x}(t)$, $t \in \left[ T_{n}, T_{n+1} \right)$.
The `rescaled $g$ iterates' are given by $\hat{g}(t) := 
\frac{\overline{g}(t)}{r(n)}$.
Let $x^n (t)$, $t \in [0,T]$ be
the solution (upto time $T$) to
$\dot{x}^{n}(t) = - \hat{g}(T_{n} + t)$, with the initial condition 
$x^{n}(0) = \hat{x}(T_{n})$, recall the definition of $\hat{g}(\cdotp)$ 
from the beginning of Section~\ref{sec:traj}. Clearly, we have
\begin{equation}
 x^{n}(t)\ = \ \hat{x}(T_{n}) - \int_{0}^{t} \hat{g}(T_{n}+z) \,dz .
\end{equation}
We begin with a simple lemma which essentially claims that
$\{x^n(t), 0 \le t \le T \mid n \ge 0\}$ =
$\{\hat{x}(T_n + t), 0 \le t \le T \mid n \ge 0 \}$. The proof is a direct consequence
of the definition of $\hat{g}$ and is hence omitted.
\begin{lemma}
\label{difftozero}
 For all $n \ge 0$, we have $x^n(t) = \hat{x}(T_n+t)$, where $t \in [0 , T]$.
\end{lemma}
It directly follows from Lemma~\ref{difftozero} that 
$\{x^n (t), t \in [0,T] \ |\ n\ge 0 \}$ = 
$\{\hat{x}(T_{n} + t), t \in [0,T]\ |\ n \ge 0\}$. 
In other words, the two families of $T$-length trajectories,
$\{x^n (t), t \in [0,T] \ |\ n\ge 0 \}$ and 
$\{\hat{x}(T_{n} + t), t \in [0,T]\ |\ n \ge 0\}$, are really one and the same.
When viewed as a subset of $C([0,T], \mathbb{R}^d)$, 
$\{x^n (t), t\in [0,T] \ |\  n \ge 0\}$ is equi-continuous and point-wise bounded. 
Further, from the \textit{Arzela-Ascoli} theorem we conclude that it is relatively compact.
In other words, $\{\hat{x}(T_{n} + t), t \in [0,T] \ |\ n \ge 0\}$
is relatively compact in $C([0,T], \mathbb{R}^{d})$.
\begin{lemma}
\label{closertoode}
 Let $r(n) \uparrow \infty$, then any limit point of $\left\{ \hat{x}(T_{n}+t), t \in [0,T] 
 : n \ge 0 \right\}$ is of the form $x(t) = x(0) +  \int_0^t g_\infty(s) \ ds$, where 
 $y: [0,T] \rightarrow \mathbb{R}^{d}$ is a measurable function and
 $g_\infty(t) \in G_{\infty}(x(t))$, $t \in [0,T]$.
\end{lemma}
\begin{proof}
For $t \ge 0$, define $[t] := max\{t(k) \ |\ t(k) \le t\}$.
Observe that for any $t \in [T_{n}, T_{n+1})$, we have $\hat{g}(t) \in G_{r(n)}(\hat{x}([t]))$
and $\lVert \hat{g}(t) \rVert$ $\le $ $K \left( 1 +  \lVert \hat{x}([t]) \rVert\right)$,
since $G_{r(n)}$ is a Marchaud map. Since $\hat{x}(\cdotp)$ is the rescaled trajectory
obtained by periodically projecting the original iterates onto a compact set, it follows that
$\hat{x}(\cdotp)$ is bounded a.s. \textit{i.e.,} 
$\underset{t \in [0, \infty)}{sup}$ $\lVert \hat{x}(t) \rVert < \infty$ $a.s.$
It now follows from the observation made earlier that $\underset{t \in [0, \infty)}{sup}$ $\lVert \hat{g}(t) \rVert < \infty$ $a.s.$

 Thus, we may deduce that there exists a sub-sequence of 
 $\mathbb{N}$, say
 $\{ l\} \subseteq \{n\}$, such that $\hat{x}(T_{l}+ \cdotp) \to x(\cdotp)$ in $C \left([0,T], 
 \mathbb{R}^{d} \right)$ and $\hat{g}(m(l)+ \cdotp) \to g_\infty(\cdotp)$ weakly in $L
 _{2} \left( [0,T], \mathbb{R}^{d} \right)$. From Lemma~\ref{difftozero}
 it follows that
 $x^{l}(\cdotp) \to x(\cdotp)$ in $C \left( [0,T], \mathbb{R}^{d} \right)$. Letting $r(l) 
 \uparrow \infty$ in
 \begin{equation}\nonumber
 x^{l}(t) = x^{l}(0) - \int_{0}^{t} \hat{g}(t(m(l) + z)) \,dz , \ t \ \in \ [0, T],
 \end{equation}
 we get $x(t) = x(0) - \int_{0}^{t} g_\infty(z) dz$ for $t \in [0,T]$. 
Since $\lVert \hat{x}(T_{n})
 \rVert \ \le \ 1$ we have $\lVert x(0) \rVert \ \le \ 1$.

 Since $\hat{g}(T_{l}+ \ \cdotp) \to g_\infty(\cdotp)$ weakly in 
 $L_{2} \left( [0,T], \mathbb{R}^{d} \right)$,
 there exists $\{l(k)\} \subseteq \{l\}$
 such that \[\frac{1}{N} \sum_{k=1}^{N} \hat{g}(T_{l(k)}+ \ \cdotp) \to g_\infty(\cdotp)
 \text{ strongly in }L_{2} \left( [0,T], \mathbb{R}^{d} \right).\]
  Further, there exists
 $\{N(m)\} \subseteq \{N\}$ such that 
 \[\frac{1}{N(m)} \sum_{k=1}^{N(m)} \hat{g}(T_{l(k)}+ \ \cdotp) \to g_\infty(\cdotp) \text{ \textit{a.e.} on } 
 [0,T].\]

 Let us fix $t_0 \in$ $\{ t \ |\ \frac{1}{N(m)} \sum_{k=1}^{N(m)} \hat{g}(T_{l(k)}+ \ t) \to g_\infty(t),\ t \in [0,T] \}$,
 then
 \begin{equation}\nonumber
  \lim_{N(m) \to \infty} \frac{1}{N(m)} \sum_{k=1}^{N(m)} \hat{g}(T_{l(k)}+ \ t_0) = g_\infty(t_0).
 \end{equation}
Since $G_{\infty}(x(t_0))$ is convex and compact (Proposition~\ref{ginfmarchaud}), 
to show that $g_\infty(t_0) \in G_{\infty}(x(t_0))$ it is enough to show
$
\underset{l(k) \to \infty}{\lim} d\left(\hat{g}(T_{l(k)} +  t_0), G_{\infty}(x(t_0))\right) = 0.
$
Suppose this is not true and $\exists$ $\epsilon > 0$ and $\{n(k)\} \subseteq \{l(k)\}$ such that
$d\left(\hat{g}(T_{n(k)} +  t_0), G_{\infty}(x(t_0))\right) > \epsilon$. Since
$\{\hat{g}(T_{n(k)} +  t_0)\}_{k \ge 1}$ is norm bounded, it follows that there is 
a convergent sub-sequence. For convenience, assume
$\underset{k \to \infty}{\lim}$ $ \hat{g}(T_{n(k)} +  t_0)  = g_0$, for some
$g_0 \in \mathbb{R}^{d}$. Since $\hat{g}(T_{n(k)} +  t_0) \in G_{r(n(k))}(\hat{x}([T_{n(k)} +  t_0]))$
and $\underset{k \to \infty}{\lim}$ $ \hat{x}([T_{n(k)} +  t_0])  = x(t_0)$, it
follows from assumption $(A4)$ that $g_0 \in G_{\infty}(x(t_0))$. 
This leads to a contradiction.
\end{proof}
 Note that in the statement of Lemma~\ref{closertoode} 
 we can replace `$r(n) \uparrow \infty$' by `$r(k) \uparrow \infty$',
 where $\{ r(k)) \}$ is a subsequence of $\{ r(n) \}$.
 Specifically we can conclude that
 any limit point of $\left\{ \hat{x}(T_{k}+t), t \in [0,T] \right\} _{
  \{k\} \subseteq \{ n \} }$ in $C([0,T], \mathbb{R}^d)$, conditioned on $r(k) \uparrow \infty$, 
  is of the form $x(t) = x(0) -  \int_0^t g_\infty(z) \,dz$, where $g_\infty(t) \in
 G_{\infty}(x(t))$ for $t \in [0,T]$. It should be noted that $g_\infty(\cdotp)$
 may be sample path dependent (if $\epsilon_n$ is stochastic then $g_\infty(\cdotp)$
 is a random variable). 
 Recall that $\underset{x \in \mathcal{A}}{sup} \ \lVert x \rVert $ = 
 $\delta_1 < \delta_2 < \delta_3 < \delta_4 < a$ 
 (see the sentence following $(A3)$ in Section~\ref{sec:assumptions}).
 The following is an immediate corollary of Lemma~\ref{closertoode}.
 \begin{corollary} \label{r_0}
  $\exists \ 1 < R_0 < \infty$ 
such that $\forall \ r(l) > R_0$,
$\lVert \hat{x}(T_l + \cdotp) - x(\cdotp) \rVert < \delta_3 - \delta_2$,
where $\{ l \} \subseteq \mathbb{N}$ and $x(\cdotp)$ is a solution (up to time $T$) of 
$\dot{x}(t) \in -G_\infty(x(t))$
such that $\lVert x(0) \rVert \le 1$. The form of $x(\cdotp)$ is as given by
Lemma~\ref{closertoode}.
 \end{corollary}
\begin{proof}
Assume to the contrary that $\exists \ r(l) \uparrow \infty$ such that
$\hat{x}(T_l + \cdotp)$ is at least $\delta_3 - \delta_2$ away from any solution
to the $DI$. It follows from Lemma~\ref{closertoode} that
there exists a subsequence of 
$\{ \hat{x}(T_l + t), 0 \le t \le T \ :\ l \subseteq \mathbb{N} \}$
guaranteed to converge, in $C([0,T], \mathbb{R}^d)$,
to a solution of $\dot{x}(t) \in -G_\infty (x(t))$ such that $\lVert x(0) \rVert \le 1$. 
This is a contradiction.
 \end{proof}
\begin{remark} \label{r0}
 It is worth noting that $R_0$ may be sample path dependent.
Since $T = T(\delta_2 - \delta_1) +1$ we get $\lVert \hat{x}([T_l + T]) \rVert < \delta_3$
for all $T_l$ such that $\lVert \overline{x}(T_l)\rVert (=r(l)) > R_0$.
\end{remark}
\subsection{Main Results} \label{sec:main}
We are now ready to prove the two main results of this paper. We begin by showing that
(\ref{eq:BP2}) is stable (bounded a.s.). In other words, we show that $\underset{n}{sup}
\lVert r(n) \rVert < \infty$ a.s. Once we show that the iterates are stable we use the main
results of Bena\"{i}m, Hofbauer and Sorin to conclude that the iterates converge
to a closed, connected, internally chain transitive and invariant
set of $\dot{x}(t) \in G(x(t))$.
\begin{theorem}
\label{stability}
  Under assumptions $(A1) - (A4)$, the iterates given by (\ref{eq:BP2}) are stable
  \textit{i.e.,} $\underset{n}{sup} \lVert x_{n} \rVert < \infty$ \textit{a.s.}
  Further, they converge to a closed, connected, internally chain transitive and invariant
set of $\dot{x}(t) \in G(x(t))$.
\end{theorem} 
\begin{proof} 
 First, we show that the iterates are stable. To do this we start by assuming the negation
 \textit{i.e.,} $P( \underset{n}{\sup} \ r(n) = \infty) > 0$. 
 Clearly, there exists $\{l\} \subseteq \{n\}$ such that $r(l) \uparrow \infty$.
 Recall that $T_l = t(m(l))$ and that $[T_{l}+T] = max\{t(k) \ |\ t(k) \le T_{l}+T \}$.
 
We have $\lVert x(T) \rVert < \delta_2$ since $x(\cdotp)$ is a solution, up to time $T$, 
to the $DI$ given by $\dot{x}(t) \in G_{\infty}(x(t))$ and $T = T(\delta_2 - \delta_1) + 1$.
Since the rescaled trajectory is obtained by projecting onto a compact set, it follows that
the trajectory is bounded. In other words, $\underset{t \ge 0}{\sup} \ \lVert \hat{x}(t) \rVert \le K_w < \infty$,
where $K_w$ could be sample path dependent. Now,
we observe that there exists $N$ such that all of the following happen: \\
(i) $m(l) \ge N$ $\implies$  $r(l) \ > \ R_0$. [since $r(l) \uparrow \infty$]\\ 
(ii)
$m(l) \ge N$ $\implies$ $\lVert \hat{x}([T_{l}+T]) \rVert < \delta_3$.
[since $r(l) \ > \ R_0$ and Remark~\ref{r0}] \\
(iii)  $n \ge N$ $\implies$
$\gamma (n) < \frac{\delta_4 - \delta_3}{K(1+K_{\omega})}$. [since $\gamma(n) \to 0$]

We have
$\underset{x \in \mathcal{A}}{\sup} \ \lVert x \rVert = \delta_1 < \delta_2 < \delta_3 < \delta_4 <a$ 
(see the sentence following $(A3)$ in Section~\ref{sec:assumptions} for more details).
Let $m(l) \ge N$ and $T_{l+1}=t(m(l+1)) \ = \ t(m(l) + k + 1)$ for some $k > 0$.
If $T_l + T \neq T_{l+1}$ then
$t(m(l)+k) \ = \ \left[ T_{l}+T \right]$, else if $T_l + T = T_{l+1}$
then $t(m(l)+k+1) \ = \ \left[ T_{l}+T \right]$. We proceed assuming that 
$T_l + T \neq T_{l+1}$ since the other case can be identically analyzed.
Recall that $\hat{x}(T_{n+1}^{-}) \ =$ $\lim_{t \uparrow t(m(n+1))} \hat{x}(t)$, $t \in 
\left[T_{n}, T_{n+1} \right)$ and $n \ge 0$.
Then,
\begin{equation} \nonumber
 \hat{x}(T_{l+1}^{-})\ = \ \hat{x}(t(m(l)+k)) \ - \gamma (m(l)+k) 
\hat{g}(t(m(l)+k)).
\end{equation}
Taking norms on both sides we get,
\begin{equation} \nonumber
\lVert \hat{x}(T_{l+1}^{-}) \rVert\ \le \  
\lVert \hat{x}(t(m(l)+k)) \rVert \ + 
\gamma (m(l)+k) \lVert \hat{g}(t(m(l)+k)) \rVert. 
\end{equation} 
As a consequence of the choice of $N$ we get:
\begin{equation}
 \lVert \hat{g}(t(m(l)+k)) \rVert \ \le \ K \left( 1 + \lVert \hat{x}(t(m(l)+k) \rVert \right)
 \le \ K \left( 1 + K_{\omega} \right).
\end{equation}
Hence,
\begin{equation}\nonumber
\lVert \hat{x}(T_{l+1}^{-}) \rVert \ \le \ 
\lVert \hat{x}(t(m(l)+k))\rVert \ + \gamma (m(l)+k) K(1+K_{\omega}).
\end{equation}
In other words, $\lVert \hat{x}(T_{l+1}^{-}) \rVert \ < \ \delta_4$. Further,
\begin{equation} \label{eq:mainthm}
 \frac{\lVert \overline{x}(T_{l+1}) \rVert}{\lVert \overline{x}(T_{l}) 
 \rVert} \ = \ \frac{\lVert \hat{x}(T_{l+1}^{-}) \rVert}{\lVert \hat{x}(T_{l}) 
 \rVert} \ < \frac{\delta_4}{a} < 1. 
\end{equation}

It follows from (\ref{eq:mainthm}) that $\lVert \overline{x}(T_{n+1}) \rVert < 
\frac{\delta_4}{a} \lVert \overline{x}(T_n) \rVert $ if $\lVert \overline{x}(T_n) \rVert > R_0$.
From Corollary~\ref{r_0} and the aforementioned we get that
the trajectory falls at an exponential rate till it enters $\overline{B}_{R_{0}}(0)$. 
Let $t \le T_{l}$, $t \in \left[T_{n}, T_{n+1}\right)$ and $n+1 \le l$, 
be the last time that $\overline{x}(t)$ jumps from within $\overline{B}_{R_{0}}(0)$ to
the outside of the ball. It follows that 
$\lVert \overline{x}(T_{n+1}) \rVert \ge \lVert \overline{x}(T_l) \rVert$.
Since $r(l) \uparrow \infty$, $\overline{x}(t)$ would be forced to make 
larger and larger jumps within an interval of length $T+1$.
This leads to a contradiction since the maximum jump size within any fixed time interval
can be bounded using the \textit{Gronwall} inequality. Thus, the iterates are shown to
be stable.

 It now follows from \textit{Theorem 3.6 \& Lemma 3.8} 
of \textbf{Bena\"{i}m, Hofbauer and Sorin} \cite{Benaim05} that the iterates
converge almost surely to a closed, connected, internally chain transitive and invariant
set of $\dot{x}(t) \in G(x(t))$.
\end{proof}
Now that the $GD$ with non-diminishing, bounded errors, given by
(\ref{eq:BP2}), is shown to be stable (bounded a.s.),
we proceed to show that these iterates in fact converge to an arbitrarily small neighborhood of
the minimum set. The proof uses \textit{Theorem 3.1}
of Bena\"{i}m, Hofbauer and Sorin \cite{Benaim12} that we state below.
First, we make a minor comment on the limiting set of GD with errors.

Recall from Remark~\ref{benaimremark}
that the chain recurrent set of $\dot{x}(t) = - \nabla f(x(t))$ is a subset of $\mathcal{M}$,
where $\mathcal{M}$ is the minimum set of $f$. We consider two cases: $(a)$ $\mathcal{M}$
is the unique global attractor of $\dot{x}(t) = - \nabla f(x(t))$; $(b)$ $\mathcal{M}$
comprises of multiple local attractors. Suppose we are in case $(a)$,
it can be shown that any compact neighborhood, $\mathcal{M} \subseteq \mathcal{K} \subset \mathbb{R}^d$, is a fundamental
neighborhood of $\mathcal{M}$. It follows from Theorem~\ref{stability}
 that the iterates are bounded almost surely. In other words, 
 $\overline{x}(t) \in \mathcal{K}_0$, $\forall$ $t \ge 0$, for some compact set $\mathcal{K}_0$,
 that could be sample path dependent, such that $\mathcal{M} \subseteq \mathcal{K}_0$. In this case, GD with errors is expected
 to converge to a small neighborhood of $\mathcal{M}$.
 Suppose we are in case $(b)$, we need to consider $\mathcal{M}' \subseteq \mathcal{M}$ such that
 the aforementioned $\mathcal{K}_0$ is a fundamental neighborhood of it. 
 In this case, GD with errors is expected
 to converge to a small neighborhood of $\mathcal{M}'$.
\\ $ \\$
 We are now ready to present \textit{Theorem 3.1}, \cite{Benaim12}.
The statement has been interpreted to the setting of this chapter for the sake of convenience.
\\
$[$\textit{\textbf{Theorem 3.1}}, \cite{Benaim12}$]$ Given $\delta > 0$, 
 there exists $\epsilon(\delta) > 0$ such that the chain recurrent set of $\dot{x}(t) = - \nabla f(x(t)) +
 \overline{B}_r(0)$ is within the $\delta$-open neighborhood of the chain recurrent set of
 $\dot{x}(t) = - \nabla f(x(t))$ for all $r \le \epsilon(\delta)$.
\begin{theorem} \label{corr}
 Given $\delta > 0$, there exists $\epsilon(\delta) > 0$ such that 
 the $GD$ with bounded errors given by (\ref{eq:BP2}) converges to $N^{\delta}(\mathcal{M})$,
 the $\delta$-neighborhood of the minimum set of $f$, provided $\epsilon < \epsilon(\delta)$.
 Here $\epsilon$ is the bound for estimation errors from assumption $(A1)$.
\end{theorem}
\begin{proof}
As stated in Remark~\ref{benaimremark}, the
chain recurrent set of $\dot{x}(t) = - \nabla f(x(t))$ is assumed to be a subset of the minimum set of $f$.
Note that the iterates given by (\ref{eq:BP2}) track a solution to $\dot{x}(t) \in - \left( \nabla f(x(t))
+ \overline{B}_\epsilon(0) \right)$. It follows from \textit{Theorem 3.1, \cite{Benaim12}} that 
(\ref{eq:BP2}) converge to a $\delta$-neighborhood of the chain recurrent set provided 
$\epsilon < \epsilon(\delta)$. In other words, GD with errors converges to a small neighborhood of the 
minimum set provided GD without errors is guaranteed to converge to the minimum set.
\end{proof}
\subsection{Implementing GD methods using SPSA} \label{implementBP}
Gradient estimators are often used in the implementation of $GD$ methods such as 
$SPSA$, \cite{spall}. When using $SPSA$ 
the update rule for the $i^{\text{th}}$ coordinate is given by
\begin{equation}
 \label{SPSA}
 x^i _{n+1} = x^i _n - \gamma(n) \left( \frac{f(x_n + c_n \Delta_n) - f(x_n - c_n \Delta_n)}{2 c_n \Delta^i _n} \right), 
\end{equation}
where $x_n = \left( x^1 _n , \dots, x^d _n \right)$ is the underlying parameter, 
$\Delta _n = \left( \Delta^1 _n , \dots, \Delta^d _n \right)$ is a sequence of 
perturbation random vectors such that $\Delta ^i _n$, $1 \le i \le d$, $n \ge 0$ are $i.i.d.$.
It is common to assume $\Delta^i_n$ to be symmetric, Bernoulli distributed, taking values $\pm 1$ $w.p.\ 1/2$. The
sensitivity parameter $c_n$ is such that the following are assumed: $c_n \to 0$ as $n \to \infty$;
$\sum _{n \ge 0} \left( \frac{\gamma(n)}{c_n} \right) ^2 < \infty$, see $A1$ of \cite{spall}. 
Further, $c_n$ needs to be chosen such that the estimation errors go to zero. This, in particular,
could be difficult since the form of the function $f$ is often unknown. One may need to run experiments
to find each $c_n$.
Also, smaller values
of $c_n$ in the initial iterates tends to blow up the variance which in turn affects convergence.
For these reasons, in practice, one often lets $c_n := c$ (a small constant) for all $n$. If we assume additionally
that the second derivative of $f$ is bounded, then it is easy to see that the estimation errors are
bounded by $\epsilon(c)$ such that $\epsilon(c) \to 0$ as $c \to 0$. 
\textit{\textbf{Thus, keeping $c_n$ fixed to $c$ forces the estimation errors to be bounded at each stage.
In other words, SPSA with a constant sensitivity parameter falls under the purview of the framework presented in this paper.}}
Also, it is worth noting that the iterates
are assumed to be stable (bounded a.s.) in \cite{spall}. However in our framework,
stability is shown under verifiable conditions even when $c_n = c, \ n \ge 0$.
\\ \indent
We arrive at the important question of how to choose this constant $c$ in practice such that fixing $c_n := c$
we still get the following: (a) the iterates are stable and (b) $GD$ implemented in this manner converges to
the minimum set. Suppose the simulator wants to ensure that the iterates converge to a $\delta$-neighborhood
of the minimum set \textit{i.e.,} $N^\delta (\mathcal{M})$, then it follows from Theorem~\ref{corr}
that there exists $\epsilon(\delta)>0$ such that the $GD$ converges to $N^\delta (\mathcal{M})$ provided
the estimation error at each stage is bounded by $\epsilon(\delta)$. Now,
$c$ is chosen such that $\epsilon(c) \le \epsilon(\delta)$. 
The simulation is carried out by fixing the sensitivity parameters to this $c$. 
As stated earlier one may need to carry out 
experiments to find such a $c$. However, the advantage is that we only need to do this once
before starting the simulation. Also, the iterates are guaranteed to be stable and converge
to the $\delta$-neighborhood of the minimum set provided $(A1)$-$(A4)$ are satisfied.
\section{Experimental results} \label{experiments}
The experiments presented in this section consider a
quadratic objective function $f: \mathbb{R}^d \to \mathbb{R}$ with $f(x):= x^T Q x$, where
$Q$ is a positive definite matrix.
The origin is the unique global minimizer of $f$. On the other hand, if one were to conduct
these experiments using $f$ with
multiple local minima, then their results are expected to be similar.
\subsection{Exp.1: SPSA with constant sensitivity parameters (SPSA-C)}
\begin{figure}
\centering
\includegraphics[width=15cm, height=17cm]{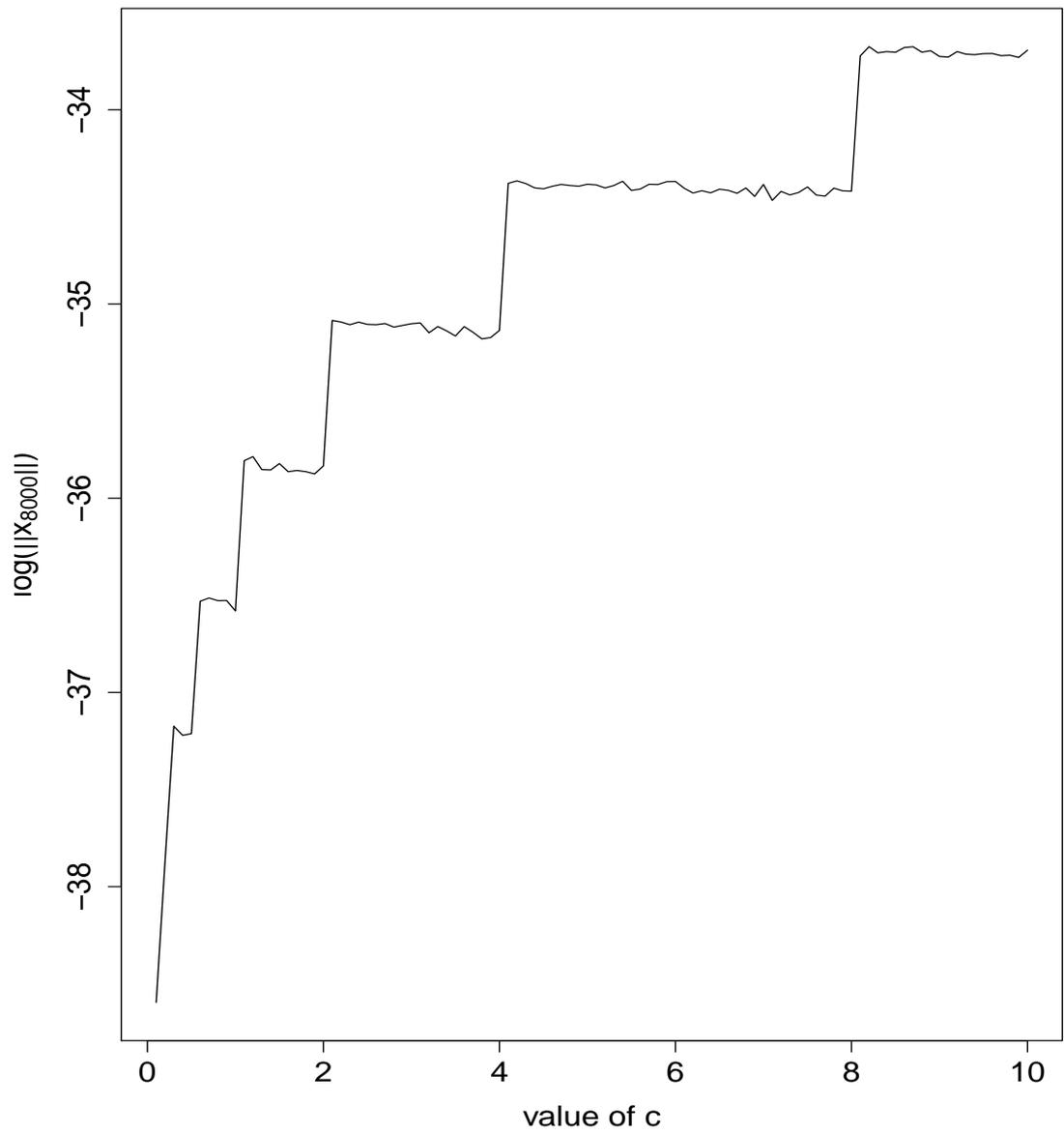}
\caption{Average performance variation of 20 independent simulation runs as a function of the sensitivity parameter $c$.} 
\label{fig1}
\end{figure}
\begin{figure}
\centering
\includegraphics[width=15cm, height=17cm]{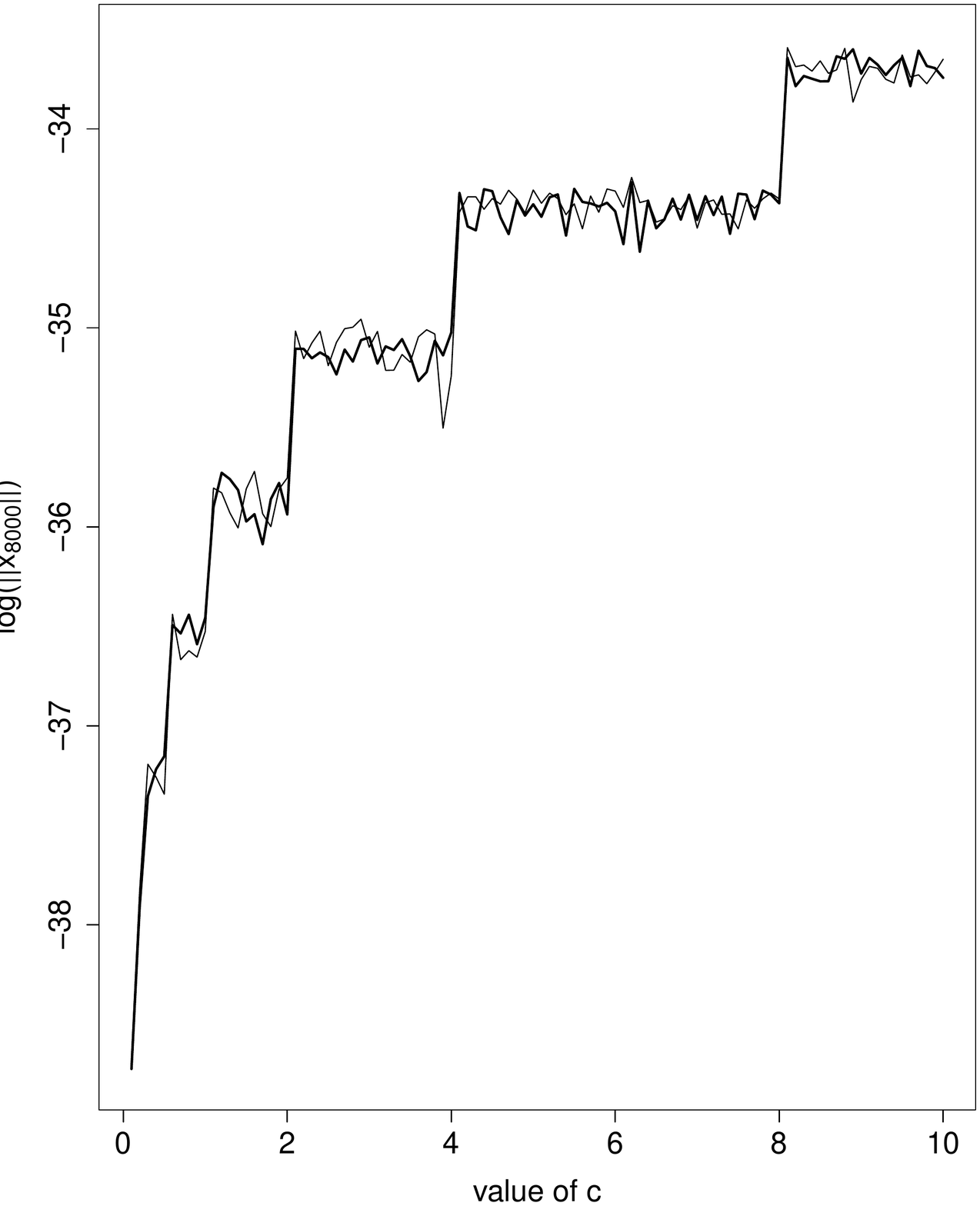}
\caption{Two sample runs.} 
\label{fig1.1}
\end{figure}
First we consider $SPSA$ with constant sensitivity parameters to find the minimum set of $f$. 
This scheme is given by (\ref{SPSA}) but with $c_n = c$ for all $n$, and we refer to it
as SPSA-C.

Parameter settings:\\
\textbf{(1)}  \textit{The positive definite matrix $Q$ and the starting point $x_0$ were randomly chosen.}
 \textbf{(2)} The dimension $d=10$. The number of iterations of SPSA-C was $8000$.
 \textbf{(3)} $c$ was varied from $0.1$ to $10$. For each value of $c$, SPSA-C
 was run for $8000$ iterations and $ \lVert x_{8000} \rVert$ was recorded. Since
 origin is the unique global minimizer of $f$, $ \lVert x_{8000} \rVert$ records
 the distance of the iterate after $8000$ iterations from the origin.
 \textbf{(4)} For $0 \le n \le 7999$, we chose the following step-size sequence:
 $a(n) = \frac{1}{(n \mod 800) + 100}$, $n \ge 1$.
This step-size sequence seems to expedite the convergence 
of the iterates to the minimum set. We were able to use this sequence
since our framework does not impose extra restrictions
on step-sizes, unlike \cite{spall}.

Since we keep the sensitivity parameters fixed the implementation was greatly simplified.
Based on the theory presented in this paper,
for larger values of $c$ one expects the iterates to be farther from the origin than for smaller values of $c$. 
This theory is corroborated by the experiment illustrated in Fig.~\ref{fig1}.

Note that to generate $Q$ we first randomly generate a column-orthonormal matrix $U$ and let $Q:=U \Sigma U ^T$,
where $\Sigma$ is a diagonal matrix with strictly positive entries. To generate
$U$, we sample it's entries independently from a Gaussian distribution and then 
apply Gram-Schmidt orthogonalization 
to the columns.

Fig.~\ref{fig1} shows the average performance of $20$ independent simulation 
runs (for each $c$) of the experiment, where $Q$
and $x_0$ were randomly chosen for each run; Fig.~\ref{fig1.1} shows two sample runs.
In Fig.~\ref{fig1} and~\ref{fig1.1} the $x$-axis represents the values of $c$ ranging from $0.1$
to $10$ in steps of $0.01$. The $y$-axis in Fig.~\ref{fig1} represents the
\textit{logarithm of the average of corresponding distances from the origin
after $8000$ iterations} \textit{i.e.,} $\log \left( 1/20\sum_{i=1}^{20} \lVert x^i_{8000} \rVert \right)$,
where $x^i_{8000}$ is the iterate-value after $8000$ runs from the $i^{th}$ simulation.
The $y$-axis in Fig.~\ref{fig1.1} represents the \textit{logarithm of the corresponding distances from the origin
after $8000$ iterations} \textit{i.e.,} $\log (\lVert x_{8000} \rVert )$. 
Note that for
$c$ close to $0$, $ x_{8000} \in B_{ e^{-38}}(0)$ 
while for $c$ close to $10$,
$ x_{8000} \in B_{e^{-32}}(0)$ only.
Also note that the graph has a series of  ``steep rises'' followed by ``plateaus''. These indicate that
for values of $c$ within the same plateau the iterate converges to the same neighborhood of
the origin. As stated earlier for larger values of $c$ the iterates are farther from the origin than
for smaller values of $c$.
\subsection{Exp.2: GD with constant gradient errors}
For the second experiment we ran the following recursion for $1000$ iterations:
\begin{equation}\label{ex:2}
 x_{n+1} = x_n + 1/n \left( Q x_n + \underline{\epsilon} \right), \text{ where}
\end{equation}
 \textbf{(a)} the starting point $x_0$ was randomly chosen and dimension $d=10$.
 \textbf{(b)} the matrix $Q$ was a randomly generated positive definite matrix ($Q$ is generated as explained before).
 \textbf{(c)} $\underline{\epsilon} = \left(
                                \epsilon / \sqrt{d}
                                \ldots 
                                \epsilon / \sqrt{d}
                               \right)
$, is the constant noise-vector added at each stage and $\epsilon \in \mathbb{R}$.

Since $Q$ is positive definite, we expect (\ref{ex:2}) to converge to the origin
when $\epsilon = 0$ in the noise-vector. 
A natural question to ask is the following: If a ``small'' noise-vector is added at each stage 
does the iterate sequence still
converge to a small neighborhood of the origin or do the iterates diverge?
It can be verified that (\ref{ex:2}) satisfies $(A1)$-$(A4)$
of Section~\ref{sec:assumptions} for any $\epsilon \in \mathbb{R}$. Hence it follows from
Theorem~\ref{stability} that the iterates are stable and do not diverge.
In other words, the addition of such a noise does not accumulate and force the iterates to diverge. As in
the first experiment we expect the iterates to be farther from the origin for larger values of $\epsilon$.
This is evidenced by the plots in Fig.~\ref{fig2} and~\ref{fig2.1}. 
\begin{figure}
\centering
\includegraphics[width=15cm, height=17cm]{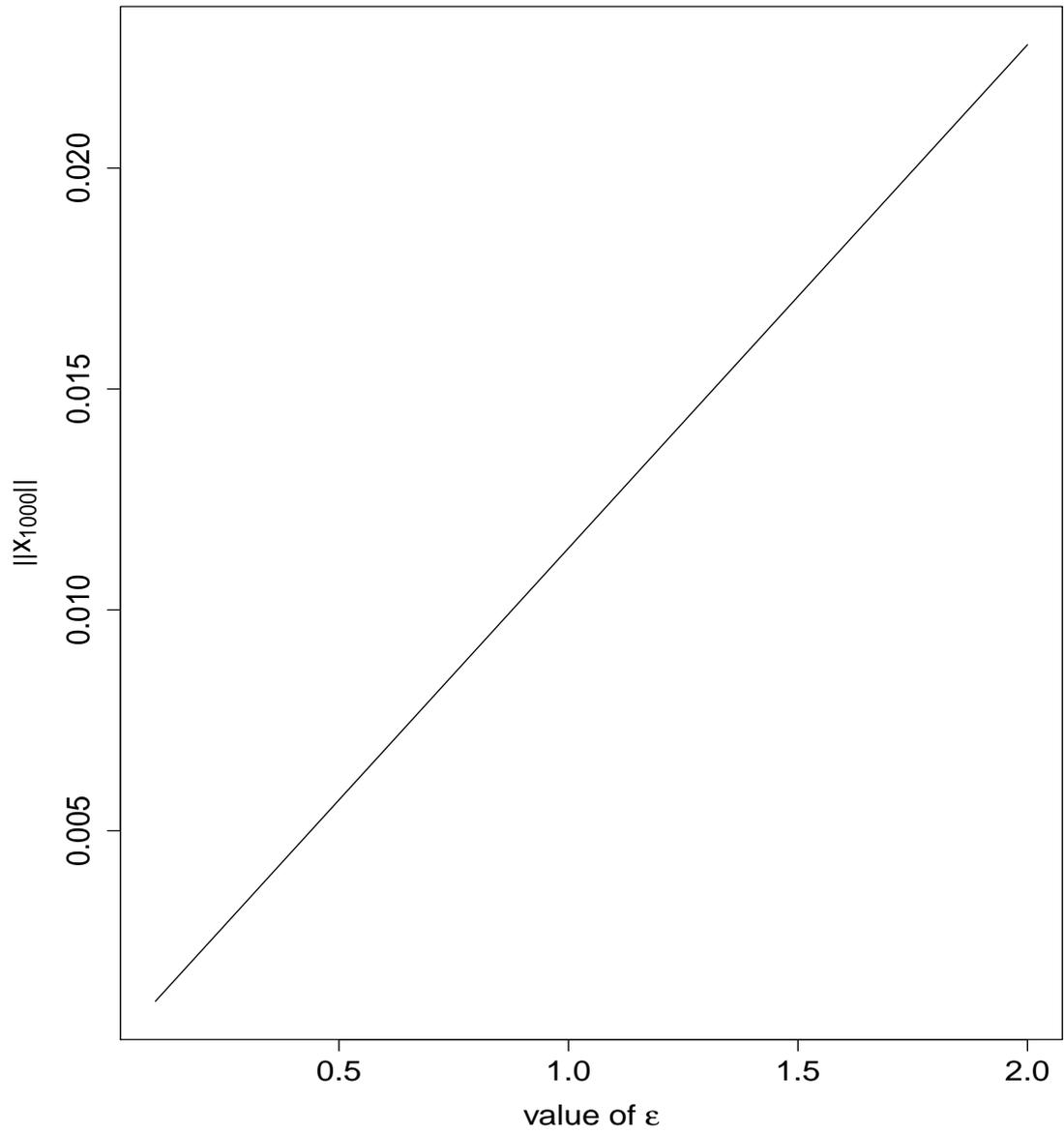}
\caption{Average performance variation of 20 independent simulation runs as a function of the neighborhood parameter $\epsilon$.} 
\label{fig2}
\end{figure}
As before, Fig.~\ref{fig2} shows the 
average performance of $20$ independent simulation runs (for each $\underline{\epsilon}$) and Fig.~\ref{fig2.1} shows three of these sample runs.
The $x$-axis in Fig.~\ref{fig2} and~\ref{fig2.1} represents values of the $\epsilon$ parameter in (\ref{ex:2}) 
that varies from $0.1$ to $2$ \textit{i.e.,}
$\lVert \underline{\epsilon} \rVert$ varies from $0.1$ to $2$ in steps of $0.01$. 
The $y$-axis in Fig.~\ref{fig2} represents the average distance of the iterate
from the origin after $1000$ iterations \textit{i.e.,} $1/20 \sum_{i=1}^{20}\lVert x^i_{1000} \rVert$,
where $x^i_{1000}$ is the iterate-value after $1000$ iterations from the $i^{th}$ run. 
The $y$-axis in Fig.~\ref{fig2} represents $\lVert x^i_{1000} \rVert$. 
For $\epsilon$ close to $0$
the iterate (after $1000$ iterations) is within $B_{0.0003}(0)$ while for $\epsilon$ close to $2$ 
the iterate (after $1000$ iterations) is only within $B_{0.1}(0)$.
\begin{figure}
\centering
\includegraphics[width=15cm, height=17cm]{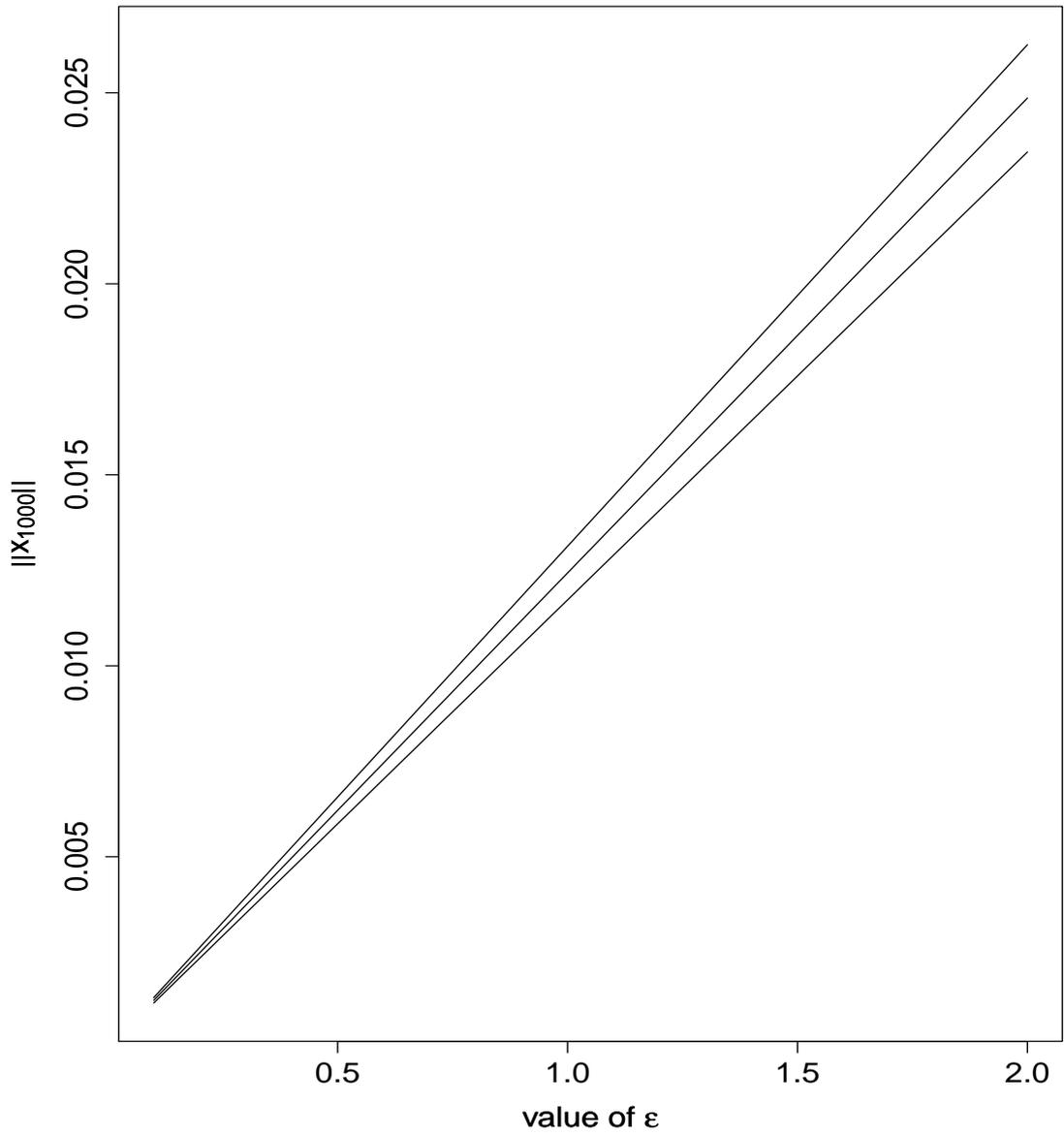}
\caption{Three sample runs.} 
\label{fig2.1}
\end{figure}
\newpage
\section{Extensions and conclusions}\label{extensions} 
In this paper we have provided sufficient conditions for stability and convergence
(to a small neighborhood of the minimum set) of $GD$ with bounded and (possibly) non-diminishing 
errors. 
To the best of our knowledge our analysis of $GD$
with errors is new to the literature. In addition to being easily verifiable,
the assumptions presented herein do not affect the choice of step-size. Finally, experimental results
presented in Section~\ref{experiments} are seen to validate the theory. 
\textit{An important step in the analysis of `$GD$ with errors' is to show stability (almost sure boundedness)
of the iterates. It is worth noting that
this step is not straightforward even in the case of asymptotically vanishing
errors, \textit{i.e.,} $\epsilon_n \to 0$ as $n \to \infty$.
}
An extension to our main results is the
introduction of an additional martingale noise term $M_{n+1}$ at stage $n$. Our results 
will continue to hold provided $\sum_{n \ge 0} \gamma(n) M_{n+1} < \infty$ a.s.
Another extension is to analyze implementations of 
 $GD$ using Newton's method with bounded, (possibly) non-diminishing errors.
 To see this, define $G(x) := H(x)^{-1} \nabla f(x) + \overline{B}_\epsilon(0)$ in $(A1)$; $G_\infty$
 changes accordingly. 
 Here $H(x)$ (assumed positive definite) denotes the Hessian evaluated at $x$.
 Theorems~\ref{stability} \& \ref{corr} hold under this new definition of $G$
 and appropriate modifications of $(A1)-(A4)$.
Our analysis is valid in situations where the function $f$ is not differentiable at some points, however,
the error in the gradient estimate  at any stage is bounded.
An interesting
future direction will be to derive convergence rates of gradient schemes with 
non-diminishing errors.
More generally, it would be interesting to derive convergence rates of stochastic approximation algorithms with set-valued
mean-fields.
\bibliographystyle{plain}
\bibliography{example_paper}
\end{document}